\documentclass[%
 reprint,
 superscriptaddress,
 longbibliography,
 amsmath,amssymb,
 aps,
 pra,
]{revtex4-1}

\usepackage[utf8]{inputenc} 
\usepackage{amsmath,amsfonts,amssymb,mathrsfs,amsthm,cases} 
\usepackage{bm, bbm, dsfont} 
\usepackage{physics} 
\usepackage{graphicx} 
\usepackage{diagbox,threeparttable,dashrule,booktabs,dcolumn} 
\usepackage{xcolor} 
\usepackage{url} 
\usepackage{hyperref} 
\hypersetup{colorlinks,linkcolor={blue},citecolor={blue},urlcolor={red}}
\usepackage{cleveref} 

\crefname{equation}{Eq.}{Eqs.}
\crefname{section}{Sec.}{Secs.}
\crefname{definition}{Definition}{Definitions}
\crefname{proposition}{Proposition}{Propositions}
\crefname{lemma}{Lemma}{Lemmas}
\crefname{theorem}{Theorem}{Theorems}
\crefname{corollary}{Corollary}{Corollaries}
\crefname{conjecture}{Conjecture}{Conjectures}
\crefname{claim}{}{Claims}
\crefname{example}{Example}{Examples}

\newtheorem{definition}{Definition}

\newtheorem{lemma}[definition]{Lemma}

\newtheorem{theorem}[definition]{Theorem}
\newtheorem{corollary}[definition]{Corollary}
\newtheorem{conjecture}[definition]{Conjecture}

\newtheorem{example}[definition]{Example}

\newcommand{\proj}[1]{|{#1}\rangle \langle {#1}|}

\newcommand{\nc}{\newcommand}

\def\red{\textcolor{red}}

\def\bpf{\begin{proof}}
\def\epf{\end{proof}}
\def\bea{\begin{eqnarray}}
\def\eea{\end{eqnarray}}
\def\beq{\begin{equation}}
\def\eeq{\end{equation}}
\def\bal{\begin{aligned}}
\def\eal{\end{aligned}}
\def\bma{\begin{pmatrix}}
\def\ema{\end{pmatrix}}

\def\dg{\dagger}

\def\ox{\otimes}

\def\diag{\mathop{\rm diag}}

\def\a{\alpha}
\def\b{\beta}
\def\g{\gamma}

\def\e{\epsilon}

\def\t{\theta}

\def\la{\lambda}

\def\G{\Gamma}

\nc{\bbA}{{\mathbb A}}  \nc{\bbB}{{\mathbb B}}  \nc{\bbC}{{\mathbb C}}
\nc{\bbD}{{\mathbb D}}  \nc{\bbE}{{\mathbb E}}  \nc{\bbF}{{\mathbb F}}
\nc{\bbG}{{\mathbb G}}  \nc{\bbH}{{\mathbb H}}  \nc{\bbI}{{\mathbb I}}
\nc{\bbJ}{{\mathbb J}}  \nc{\bbK}{{\mathbb K}}  \nc{\bbL}{{\mathbb L}}
\nc{\bbM}{{\mathbb M}}  \nc{\bbN}{{\mathbb N}}  \nc{\bbO}{{\mathbb O}}
\nc{\bbP}{{\mathbb P}}  \nc{\bbQ}{{\mathbb Q}}  \nc{\bbR}{{\mathbb R}}
\nc{\bbS}{{\mathbb S}}  \nc{\bbT}{{\mathbb T}}  \nc{\bbU}{{\mathbb U}}
\nc{\bbV}{{\mathbb V}}  \nc{\bbW}{{\mathbb W}}  \nc{\bbX}{{\mathbb X}}
\nc{\bbY}{{\mathbb Y}}  \nc{\bbZ}{{\mathbb Z}}  

\nc{\bA}{{\bf A}}  \nc{\bB}{{\bf B}}  \nc{\bC}{{\bf C}}
\nc{\bD}{{\bf D}}  \nc{\bE}{{\bf E}}  \nc{\bF}{{\bf F}}
\nc{\bG}{{\bf G}}  \nc{\bH}{{\bf H}}  \nc{\bI}{{\bf I}}
\nc{\bJ}{{\bf J}}  \nc{\bK}{{\bf K}}  \nc{\bL}{{\bf L}}
\nc{\bM}{{\bf M}}  \nc{\bN}{{\bf N}}  \nc{\bO}{{\bf O}}
\nc{\bP}{{\bf P}}  \nc{\bQ}{{\bf Q}}  \nc{\bR}{{\bf R}}
\nc{\bS}{{\bf S}}  \nc{\bT}{{\bf T}}  \nc{\bU}{{\bf U}}
\nc{\bV}{{\bf V}}  \nc{\bW}{{\bf W}}  \nc{\bX}{{\bf X}}
\nc{\bY}{{\bf Y}}  \nc{\bZ}{{\bf Z}}  

\nc{\cA}{{\cal A}}  \nc{\cB}{{\cal B}}  \nc{\cC}{{\cal C}}
\nc{\cD}{{\cal D}}  \nc{\cE}{{\cal E}}  \nc{\cF}{{\cal F}}
\nc{\cG}{{\cal G}}  \nc{\cH}{{\cal H}}  \nc{\cI}{{\cal I}}
\nc{\cJ}{{\cal J}}  \nc{\cK}{{\cal K}}  \nc{\cL}{{\cal L}}
\nc{\cM}{{\cal M}}  \nc{\cN}{{\cal N}}  \nc{\cO}{{\cal O}}
\nc{\cP}{{\cal P}}  \nc{\cQ}{{\cal Q}}  \nc{\cR}{{\cal R}}
\nc{\cS}{{\cal S}}  \nc{\cT}{{\cal T}}  \nc{\cU}{{\cal U}}
\nc{\cV}{{\cal V}}  \nc{\cW}{{\cal W}}  \nc{\cX}{{\cal X}}
\nc{\cY}{{\cal Y}}  \nc{\cZ}{{\cal Z}}  

\def\ox{\otimes}

\def\dg{\dagger}

\begin{document}


\title{The detection power of real entanglement witnesses under local unitary equivalence} 

\author{Yi Shen}
\email[]{yishen@jiangnan.edu.cn}
\affiliation{School of Science, Jiangnan University, Wuxi Jiangsu 214122, China}

\author{Lin Chen}
\email[Corresponding author: ]{linchen@buaa.edu.cn}
\affiliation{LMIB(Beihang University), Ministry of Education, and School of Mathematical Sciences, Beihang University, Beijing 100191, China}

\author{Zhihao Bian}
\affiliation{School of Science, Jiangnan University, Wuxi Jiangsu 214122, China}

\date{\today} 

\begin{abstract}
The imaginary unit $i$ has recently been experimentally proven to be indispensable for quantum mechanics. We study the differences in detection power between real and complex entanglement witnesses (EWs) distinguished by whether their matrix expressions incorporate imaginary parts. We show that a real EW (REW), denoted by a real Hermitian matrix, must detect one entangled state of a real density matrix, and conversely an entangled state of a real density matrix must be detected by one REW. We present a necessary and sufficient condition for the entangled states detected by REWs, and give a specific example implying the detection limitations of REWs. 
From an operational perspective, we investigate whether all entangled states are detected by the EWs locally equivalent to some REWs. We prove the validity for all NPT (non-positive partial transpose) states. We also derive a necessary and sufficient condition of the validity for the PPT (positive partial transpose) entangled states of complex density matrices. By this condition we show the validity for a family of two-qutrit PPT entangled states of rank four. Another way to figure out the problem is to check whether a counterexample exists. We propose a method to examine the existence from a set-theoretic perspective, and provide some supporting evidence of non-existence. Finally, we derive some results on local projections of EWs with product projectors. 
\end{abstract}


\maketitle


\section{Introduction}
\label{sec:intro}

Complex numbers featured by the imaginary unit $i$ not only play an essential role in mathematics, but have been taken as effective tools widely used in physics, engineering and other fields. 
In most theories of physics, e.g. electromagnetism and signal-processing theories, introducing $i$ only brings mathematical advantages, while the roots are actually formulated by real numbers, as physics experiments are expressed in terms of probabilities, hence real numbers. The quantum theory seems to be a different case, since the foundations of quantum mechanics are laid on complex numbers, e.g. complex-valued Schr\"{o}dinger equations and the operators acting on complex Hilbert spaces \cite{vonNbook,diracbook}. This has puzzled countless physicists, including the fathers of the quantum theory \cite{letters}, who preferred a real version of quantum theory. 

Then a fundamental question of whether complex numbers are actually needed in the quantum formalism is raised naturally.
For several alternative formalisms of quantum theory, it has been shown that complex numbers are not required to simulate quantum systems and their evolution \cite{rsim09,rqspra13}. 
However, a standard formalism of quantum theory has recently been confirmed to necessarily contain complex numbers, both theoretically and experimentally. In Ref. \cite{cqnat21} Renou et al. devised a Bell-like three-party game based on deterministic entanglement swapping, and predicted the game results would be different when the players obey two different formulations of quantum theory in terms of real and complex Hilbert spaces respectively. Thus the real version of quantum theory can be experimentally falsified. After that, the predictions have been successfully realized based on various experimental platforms \cite{rqtprl2201,rqtprl2202,rqtprl2203}, thus the imaginary unit $i$ is not only a mathematical tool but also an essential reality in quantum world. Moreover, the quantification of the imaginarity of quantumness was proposed by Hickey and Gour \cite{Hickey_2018}. It leads to a series of works on the measures of imaginarity of quantum states in the framework of resource theory \cite{rti21prl,rti21pra,isc23,rti24cp,bci24}. In summary, the different performances that real and complex numbers display in quantum mechanics have attracted great interest.

Previous works mentioned above inspire us to investigate how essential complex numbers are in entanglement detection. Although the separability problem can be reduced to studying the states of real density matrices for some special cases \cite{linreal13,yics19}, the entanglement theory is generally defined over the complex Hilbert space rather than the real Hilbert space. Entanglement detection is at the heart of entanglement theory, and EWs are recognized as a fundamental tool to physically detect entanglement \cite{entdect2009}. Here, we are interested in the differences of detection power between REWs and CEWs (complex EWs) whose matrices are respectively real and complex. The formulas of REWs incorporating no imaginary parts could be more useful in some practical settings according to the resource theory of imaginarity \cite{Hickey_2018}. Hence, it is necessary to characterize the detection power of REWs. The detection power of REWs is weaker than that of CEWs, if there exist entangled states detected by no REW. The existence of such states supports that entanglement theory is indeed based on complex numbers. 

The local unitary (LU) equivalence and the stochastic local operations and classical communication (SLOCC) equivalence associated with two common local operations have clear operational implications. From an operational perspective, we further consider the detection power of REWs under such two local equivalences. 
The LU and SLOCC orbits are widely used in the classification of kinds of operators including entangled states \cite{3qubitlu2000,3qubitinequiv2001} and non-local unitary gates \cite{decugate15,ugatesyi22}. Some essential properties of EWs remain unchanged under SLOCC equivalence, such as the inertia known as a signature of entanglement \cite{inertiays2020,inertia24}. The LU and SLOCC orbits also provide an efficient method to operationally implement some complicated EWs. For example, to implement the EWs locally equivalent to REWs, one may implement some REW first, assisted by a sequence of local operations. Thus, our work implies the efficient entanglement detection in experiments \cite{optewex18}. Moreover, similar to the idea \cite{weakent11} of evaluating the strength of entanglement in the states, we may propose a hierarchy of entangled states as: (i) the states detected by REWs $\Longrightarrow$ (ii) the states detected by some EWs LU equivalent to REWs $\Longrightarrow$ (iii) the states detected by some EWs SLOCC equivalent to REWs $\Longrightarrow$ (iv) all entangled states.

In this paper, we first analyze the differences in detection power between REWs and CEWs. 
In Theorem \ref{le:realEW}, we show that an entangled state of a real density matrix must be detected by one REW, and present a necessary and sufficient condition to determine whether an entangled state of a complex density matrix is detected by an REW. This condition completely relies on the separability of the real part of the complex density matrix. Using this condition we propose a specific example that cannot be detected by any REW. It implies the detection limitation of using REWs only. Then, we study the detection power of REWs under local equivalences. The problem is formulated by Conjecture \ref{cj:realew} which says that every entangled state is detected by some EW locally equivalent to a real one. We focus on the case of LU equivalence, i.e. Conjecture \ref{cj:realew} (i), and generalize some results to the case of SLOCC equivalence, i.e. Conjecture \ref{cj:realew} (ii). By Lemma \ref{le:cj1-c1} we show that Conjecture \ref{cj:realew} is valid for all NPT states, and thus restrict the study on PPT entangled states of complex density matrices. We also present conditions in Lemma \ref{le:cj1-c1} to verify whether such states are detected by the REWs under LU equivalence. By summarizing the above results, we present a flowchart, Fig. \ref{fig:flow}, to examine whether an entangled state is detected by an EW LU equivalent to some real one. Following this flowchart, we prove Conjecture \ref{cj:realew} for a family of two-qutrit states by Theorem \ref{thm:PPTEr4}. In Example \ref{ex:pptreal} we also construct a two-quqart PPT entangled state detected by no REW but can be detected by an EW LU equivalent to a real one. Furthermore, we propose an equivalent method to examine the existence of a counterexample against Conjecture \ref{cj:realew} from a set-theoretic perspective. According to Lemma \ref{le:prs-prop}, the counterexamples share several interesting properties with separable states, which suggests that the counterexamples may not exist in the low-dimensional systems. Finally, to link the EWs supported on high-dimensional spaces to those supported on lower-dimensional spaces, we investigate if it is possible to locally project an EW to another in Lemma \ref{le:cj2-c1}. As a byproduct, by Lemma \ref{le:projmindim} we estimate how small the local dimensions of the spaces supporting the locally projected EWs are.

The remainder of this paper is organized as follows: In Sec. \ref{sec:pre} we clarify some notations and definitions, and present the known facts and lemmas as the necessary basics to investigate the focused problems. In Sec. \ref{sec:limit} we reveal the differences in detection power between REWs and CEWs, and point out the limitations by using REWs only. In Sec. \ref{sec:detpower} we further study the detection power of REWs under two local equivalences. The problem is specifically formulated as Conjecture \ref{cj:realew} at the beginning of this section. In Sec. \ref{sec:ppt+sepr} we study Conjecture \ref{cj:realew} by contradiction, and characterize the counterexamples against the conjecture, from a set-theoretic perspective. In Sec. \ref{subsec:cj2} we consider whether an EW can be locally projected to an EW supported on a lower-dimensional space. Finally, the concluding remarks are given in Sec. \ref{sec:con}.

\section{Preliminaries}
\label{sec:pre}

There are two parts of this section. In the first part, Sec. \ref{subsec:nandd}, we clarify some notations and definitions as the basics. In the second part, Sec. \ref{subsec:symmat}, we first introduce some necessary facts on the symmetric and skew-symmetric matrices, as the real and imaginary parts of a Hermitian matrix are symmetric and skew-symmetric respectively. Next, we present some known lemmas which are useful to the problems studied in this paper.

\subsection{Notations and Definitions}
\label{subsec:nandd}

We first clarify the LU and SLOCC equivalences. For two bipartite matrices $M,N$, they are LU equivalent, denoted by $M\sim_{LU} N$, if there exists a product unitary matrix $X=U\ox V$ such that $XMX^\dg=N$. If such $U$ and $V$ of $X$ are generalized to be invertible, then $M,N$ are precisely SLOCC equivalent, denoted by $M\sim_{SLOCC} N$. Both LU operations and SLOCC do not change the separability of a state. Thus, such two local equivalences are widely used in characterizing entanglement structures. 

Our main purpose is to analyze the differences in detection power under local equivalences between REWs and CEWs whose matrices are respectively real and complex. First of all, it is necessary to specify the number fields for matrices.
Denote by $\cM_{m,n}(\bbC)$ and $\cM_{m,n}(\bbR)$ the two sets of matrices with $m$ rows and $n$ columns over the complex and real fields respectively. The matrices in the two sets above are called complex and real matrices respectively. For simplicity, if $m=n$, we may write the two sets as $\cM_{n}(\bbC)$ and $\cM_{n}(\bbR)$. For any matrix $M\in\cM_{m,n}(\bbC)$, denote by $M^*$ and $M^\dg$ the complex conjugate and conjugate transposition of $M$ respectively. Based on the conjugate transposition we denote $\abs{M}:=\sqrt{M^\dg M}$ which is Hermitian. For any Hermitian $H$, we call $H^+:=\frac{1}{2}(H+H^*)$ and $H^-:=\frac{1}{2i}(H-H^*)$ respectively the real and imaginary parts of $H$, according to the unique decomposition of Hermitian matrices as $H=H^+ + i H^-$ where $H^+$ is real symmetric and $H^-$ is real skew-symmetric.

We focus on bipartite entanglement. It is thus necessary to clarify the notations about bipartite Hermitian matrices. For any Hermitian matrix $H$, denote by $\cR(H)$ and $r(H)$ the range and the rank of $H$ respectively. For convenience, we simply denote the bipartite Hermitian matrix $M_{AB}$ as an $a\times b$ matrix, if $r(\tr_B(M_{AB}))\leq a$ and $r(\tr_A(M_{AB}))\leq b$. For a bipartite matrix $M_{AB}$, denote by $M_{AB}^{\G}$ the partial transpose with respect to the first subsystem $A$. The PPT and NPT states are classified by whether their partial transposes are still positive semidefinite. Note that in Ref. \cite{rebit2001} the authors considered the so-called ``real entanglement'' defined in vector spaces over the real numbers, which is different from the standard quantum entanglement. In this paper we consider entanglement in standard quantum mechanics over the complex numbers. To make it clear, \emph{a real (complex) and entangled (separable) state, in what follows, accurately means a bipartite state whose density matrix is both real (complex) and entangled (separable) by standard definition over the complex Hilbert space.} 

Next, we introduce some sets of EWs which are frequently used in the following sections. Note that here we shall omit ``bipartite'' before EWs, unless otherwise noted. A main problem in this paper asks which entangled states can be detected by the EWs locally equivalent to some REWs. To figure out the problem, we introduce the corresponding sets of EWs as follows. It is worth mentioning that the local equivalents of an EW are still EWs for both LU and SLOCC equivalences, which makes Definition \ref{def:ewsets} is well-defined.
\begin{definition}
    \label{def:ewsets}
    We define $\cE^{LU}$ as the set of EWs which are LU equivalent to some REWs, and similarly $\cE^{SLOCC}$ as the set of EWs which are SLOCC equivalent to some REWs.
\end{definition}

For a given entangled state $\rho$, to determine whether it is detected by some EW in $\cE^{LU}$ or $\cE^{SLOCC}$, we may consider whether such two sets intersect with the set $\cE_{\rho}$,
\beq
    \label{eq:ewrho}
    \cE_{\rho}:=\{W~|~\text{$W$ is an EW and $\tr(W\rho)<0$.}\}.
\eeq
The set $\cE_{\rho}$ formulated by Eq. \eqref{eq:ewrho} is the set including all EWs which detect $\rho$. 

Due to the cyclicity law of trace, i.e. $\tr(AB)=\tr(BA)$, the local operators acting on EWs can be transferred to acting on states, namely $\tr((U\ox V)W(U\ox V)^\dg\rho)=\tr(W(U\ox V)^\dg\rho(U\ox V))$ for an EW $W$ and a state $\rho$. Hence, the characterization of $\cE^{LU}$ can be transformed to characterizing a special set of PPT states defined below. The process of transformation will be discussed in detail in Sec. \ref{sec:ppt+sepr}.

\begin{definition}
    \label{def:pptseprset}
Denote by ${\cal P}_{rs}(m,n)$ the set of states $\rho_{AB}$ supported on $\bbC^m\ox\bbC^n$, satisfying the following conditions: 

(i) $\rho_{AB}$ is a PPT state;

(ii) there exists some $\sigma_{AB}\sim_{LU}\rho_{AB}$ such that the real part of $\sigma_{AB}$, namely $\sigma_{AB}^+$, is separable.
\end{definition}

Denote by $\cS(m,n)$ the set of all separable states supported on $\bbC^m\ox\bbC^n$. One can verify directly by Definition \ref{def:pptseprset} that all separable states supported on $\bbC^m\ox\bbC^n$ are included in $\cP_{rs}(m,n)$, i.e. $\cS(m,n)\subseteq\cP_{rs}(m,n)$.

\subsection{Known Facts and Lemmas}
\label{subsec:symmat}

As the real and imaginary parts of a Hermitian operator are symmetric and skew-symmetric respectively, we first give the obvious facts on symmetric and skew-symmetric matrices. 
Suppose that $S,A\in\cM_n(\bbR)$ are respectively symmetric and skew-symmetric, and $\ket{a},\ket{b}\in\bbR^n$. By definition it follows that
\begin{eqnarray}
&&
\bra{a}S\ket{b}=\bra{b}S\ket{a},
\\&&
\bra{a}A\ket{a}=
\bra{a}A\ket{b}+\bra{b}A\ket{a}=0.
\end{eqnarray}

Second, we are interested in local equivalences, especially the LU equivalence, regarding to entangled states and EWs. For this reason, we introduce a necessary decomposition of unitary matrices by virtue of real orthogonal matrices as below.

\begin{lemma}
    \label{le:unitarydec}
For any unitary matrix $U\in\cM_n(\bbC)$, there are two real orthogonal matrices $V_1, V_2$, and a diagonal matrix $D\in\cM_n(\bbC)$ such that $U=V_1DV_2$.
\end{lemma}

This decomposition can be derived from a known result \cite[2.5.P57]{bookmatrix} that a square matrix $M$ is normal and symmetric if and only if there is a real orthogonal matrix $Q$ and a diagonal matrix $\Lambda$ such that $M=Q^T\Lambda Q$. Applying this result to $U^T U$ for any unitary $U$, we obtain that there is a real orthogonal matrix $V_2$ and a diagonal matrix $\Lambda$ such that $U^T U=V_2^T\Lambda V_2$. Denote by $D$ the square root of $\Lambda$. Let $U=V_1DV_2$ for some unitary $V_1$. It follows from the equality $U^TU=V_2^T D^2 V_2=V_2^T DV_1^T V_1 D V_2$ that $V_1^TV_1=I$. Since $V_1$ is a unitary matrix satisfying $V_1^\dg V_1=I$, we conclude that $V_1$ is a real orthogonal matrix due to $V_1^\dg=V_1^T$. Thus, we obtain $U=V_1DV_2$ for real orthogonal $V_1,V_2$ and diagonal $D$.
Lemma \ref{le:unitarydec} shows that the essence of a complex unitary matrix is reflected by the complex diagonal entries of the diagonal matrix $D$ in the above decomposition.

Third, we introduce some known results on EWs which are used to construct REWs from CEWs. Recall that a Hermitian $W$ is called an EW, if $\tr(W\rho)\geq 0$ for any separable state $\rho$, and $\tr(W\sigma)< 0$ for at least one entangled state $\sigma$. To test the positivity of $\tr(W\rho)$, the following equalities are frequently used:
\begin{eqnarray}
\label{eq:treq-1}
&& \tr(MN)=\tr(N^T M^T)=\tr(M^T N^T), \\
\label{eq:treq-2}
&& (M^*)^\Gamma=(M^\Gamma)^*,~\text{and}~(M^+)^\Gamma=(M^\Gamma)^+, \\
\label{eq:treq-3}
&& \tr(MN^\Gamma)=\tr(M^\Gamma N)=\tr(NM^\Gamma),
\end{eqnarray}
where in Eqs. \eqref{eq:treq-2} and \eqref{eq:treq-3} the matrices are taken as bipartite Hermitian matrices.
It follows from Eq. \eqref{eq:treq-1} that $\tr(W\rho)=\tr(W^*\rho^*)$ for an EW $W$ and a state $\rho$. It means that the complex conjugate $W^*$ of an EW $W$ is also an EW, and $W^*$ detects the entangled state $\rho^*$ if and only if $W$ detects the entangled state $\rho$. This reveals a close connection between an EW and its complex conjugate. By virtue of this connection we may mutually calculate the two values $\tr(W\rho)$ and $\tr(W^*\rho^*)$ for different objects $W$ and $W^*$.
By applying Eq. \eqref{eq:treq-2} to PPT states, we derive a conclusion as follows. If $\rho$ is a PPT state, then for any $\sigma\sim_{SLOCC}\rho$ $(\sigma\sim_{LU}\rho)$, both $\sigma$ and its real part $\sigma^+$ are PPT states. It implies that the real parts are constrained to be PPT when considering PPT entangled states and their local equivalents. Eq. \eqref{eq:treq-2} is also useful to analyze and classify $W^+$, which will be discussed in detail by Corollary \ref{cor:W+}.
Moreover, it is known that the partial transpose of an NPT state can be regarded as an EW, which connects EWs and NPT states.
Further, the following lemma concluded from Eq. \eqref{eq:treq-3} reveals a general relation between $W$ and $W^\G$ for an EW $W$.

\begin{lemma}\cite[Lemma 7]{inertiays2020}
    \label{le:ewpt}
    A Hermitian matrix $W$ is an EW if and only if $W^\Gamma$ is an EW or an NPT state. 
\end{lemma}

\section{limited detection power of real Entanglement witnesses}
\label{sec:limit}

In this section we analyze which entangled states are detected by REWs and which entangled states can be detected only by CEWs. This question implies the detection limitations by using REWs only. To answer this question, we first characterize the set $\cE_\rho$ given by Eq. \eqref{eq:ewrho}, in Lemma \ref{le:ew_rho}. Second, by generalizing the formula as $W^+:=\frac{1}{2}(W+W^*)$, we fully characterize the properties of $tW+(1-t)W^*$, for $t\in[0,1]$, in Lemma \ref{le:WandW+}. Furthermore, in Theorem \ref{le:realEW}, we present a necessary and sufficient condition to determine whether an entangled state is detected by an REW.

It is known that any entangled state can be detected by some EW. It implies that $\cE_\rho$ is non-empty if and only if $\rho$ is entangled. By the definition of EW, for an arbitrary entangled state $\rho$, we characterize the set $\cE_{\rho}$ as follows.

\begin{lemma}
\label{le:ew_rho}
For an entangled state $\rho$, the set $\cE_{\rho}$ has the following properties:

(i) $\cE_{\rho}$ is convex. That is, $tW_1+(1-t)W_2\in\cE_{\rho},~\forall t\in[0,1]$, for any $W_1,W_2\in\cE_{\rho}$.

(ii) The complex conjugate of $\rho$, denoted by $\rho^*$, is entangled, and $W\in\cE_\rho$ if and only if $W^*\in\cE_{\rho^*}$. Furthermore, $\cE_{\rho}$ contains an REW if and only if $\cE_{\rho}$ intersects with $\cE_{\rho^*}$, i.e. $\cE_{\rho}\cap\cE_{\rho^*}\neq\emptyset$.

(iii) For PPT entangled $\rho$, $W\in \cE_{\rho}$ if and only if $W^\Gamma\in \cE_{\rho^\Gamma}$.


(iv) Suppose that $\rho$ is PPT entangled and has a separable real part, i.e., $\rho^+:=\frac{1}{2}(\rho+\rho^*)$ is separable. It follows that $\tr(W\rho^*)\geq-\tr(W\rho)>0$, if $W\in\cE_{\rho}$; and $\tr(W\rho)\geq-\tr(W\rho^*)>0$, if $W\in\cE_{\rho^*}$.

(v) Let $\sigma=(U\ox V)\rho(U\ox V)^\dg$ for some LU matrix $U\ox V$. Then $(U\ox V) W (U\ox V)^\dg\in\cE_{\sigma}$ for any $W\in\cE_\rho$. It follows that $\cE_{\sigma}=(U\ox V) \cE_\rho (U\ox V)^\dg$. 
\end{lemma}

We show the detailed proof of Lemma \ref{le:ew_rho} in Appendix \ref{sec:proof-1}.
The linear combination of $W$ and $W^*$ as $W^+:=\frac{1}{2}(W+W^*)$ provides a direct way to construct an REW from a given EW. It motivates us to characterize the properties of the general linear combination of $W$ and $W^*$ as follows. 

\begin{lemma}
    \label{le:WandW+}
Suppose that $W$ is an EW and denote $W_t:=tW+(1-t)W^*$ for $t\in[0,1]$.

(i) For any $t\in[0,1]$, $W_t$ is an EW if and only if it is not positive semidefinite.

(ii) For a given $t\in[0,1]$, $W_t$ is an EW if and only if $W$ detects $\rho_t$ with $\rho_t:=t\rho+(1-t)\rho^*$ for some state $\rho$. 
\end{lemma}

One may refer to Appendix \ref{sec:proof-1} for the detailed proof of Lemma \ref{le:WandW+}.
By virtue of Lemma \ref{le:WandW+}, we specifically classify $W^+$ as follows, relying only on the properties of $W$.
\begin{corollary}
\label{cor:W+}
Suppose that $W$ is an EW. It is known that $W^+$ is either an EW or a state. Specifically,

(i) $W^+$ is an EW if and only if $W$ detects a real and entangled state;

(ii) $W^+$ is a PPT state if and only if $W$ detects no real and entangled state and $W^\Gamma$ is either an NPT state or an EW detecting no real and entangled state;

(iii) $W^+$ is an NPT state if and only if $W$ detects no real and entangled state and $W^\Gamma$ is an EW detecting at least one real and entangled state. Moreover, if $W^+$ is an NPT state, the real and entangled states detected by $W^\Gamma$ must be NPT ones.
\end{corollary}

The proof of Corollary \ref{cor:W+} is put into Appendix \ref{sec:proof-1}.
Based on the analysis above, the real part of a CEW could be an REW. Using this fact, we derive a necessary and sufficient condition to determine whether an entangled state is detected by an REW. 

\begin{theorem}
    \label{le:realEW}
    (i) For an REW, if it detects an entangled state $\rho$, then it also detects $\rho^*$ and $\rho^+$. For a real and entangled state, if it is detected by an EW $W$, then it is also detected by the two EWs as $W^*$ and $W^+$.

    (ii) An entangled state $\rho$ is detected by some REW if and only if $\rho^+$ is a real and entangled state.
\end{theorem}

\begin{proof}
(i) Let $W$ be an REW. There exists an entangled state $\rho$ such that $\tr(W\rho)<0$. It follows that 
\beq
\label{eq:realew-1}
\tr(W\rho^*)=\tr(\rho W^*)=\tr(\rho W)<0.
\eeq
Thus, we obtain that $\tr(W\rho^+)=\tr(W\rho)<0$. In other words, both $\rho^*$ and the real state $\rho^+$ are also detected by $W$. Conversely, suppose that $\rho$ is a real and entangled state detected by some EW, namely $W$. Similarly, we conclude that
\beq
\label{eq:realew-2}
\tr(W^*\rho)=\tr(\rho^* W)=\tr(\rho W)<0.
\eeq 
We conclude from Eq. \eqref{eq:realew-2} that $\rho$ is also detected by the EW $W^*$.
By calculation, it follows that $\tr(W^+ \rho)=\tr(\rho W)<0$. It implies that $W^+$ is not positive semidefinite. According to Lemma \ref{le:WandW+} (i), we know that $W^+$ should be an EW, and thus the real and entangled state $\rho$ is detected by the REW $W^+$.

(ii) It remains to consider whether a complex and entangled state can be detected by some REW by assertion (i). We first show the ``If'' part. Let $\rho$ be a complex and entangled state whose real part $\rho^+$ is also entangled. It follows from assertion (i) that $\rho^+$ is detected by some REW, namely $W_r$. Then by Eq. \eqref{eq:realew-1} we conclude that 
\beq
\label{eq:realew-3}
\tr(W_r\rho)=\frac{1}{2}(\tr(W_r\rho)+\tr(W_r\rho^*))=\tr(W_r\rho^+)<0.
\eeq
It follows that both $\rho$ and $\rho^+$ are detected by the same REW $W_r$. Second, we show the "Only if" part. Suppose that $\rho$ is a complex and entangled state detected by some REW, namely $W_r$. According to $\rho^+=\frac{1}{2}(\rho+\rho^*)$, we assume that the state $\rho^+$ is separable by contradiction. By definition we conclude that $\tr(W\rho^+)\geq 0$ for any EW $W$. However,  one can verify as below:
\beq
\label{eq:realew-3.1}
\tr(W_r\rho^+)=\frac{1}{2}(\tr(W_r\rho)+\tr(W_r\rho^*))=\tr(W_r\rho)<0.
\eeq
The second equality holds also for Eq. \eqref{eq:realew-1}.
Then we obtain a contradiction, and thus $\rho^+$ must be entangled.

This completes the proof.
\end{proof}

We know directly from Theorem \ref{le:realEW} that (i) an REW must detect a real and entangled state, and any real and entangled state can be detected by some REW; (ii) whether an entangled state can be detected by some REW depends entirely on the separability of its real part.

Finally, by virtue of a family of entangled states, we make some necessary remarks on the results derived above. Let $\ket{\psi(\theta)}=\frac{1}{\sqrt{2}}(\ket{0,0}+e^{i\theta}\ket{1,1})$ for $\theta\in[0,2\pi)$ be the family of entangled states. First, we emphasize that the precondition that $W$ is an EW is essential to both Lemma \ref{le:WandW+} and Corollary \ref{cor:W+}, as there exists a Hermitian matrix which is not an EW but whose real part is either an EW or a state. We shall propose an example to show that a Hermitian matrix may not be an EW even if its real part is an EW. The example is given by
\beq
\label{eq:realew-4.1}
H=\frac{1}{2}
\bma
1 & i & 2i & 3i \\
-i & 0 & 1 & 4i \\
-2i & 1 & 0 & 5i \\    
-3i & -4i & -5i & 1 
\ema.
\eeq
One can verify that $H^+$ is the partial transpose of $\proj{\psi(0)}$ which is a two-qubit NPT state, and thus $H^+$ is an REW. Nevertheless, $H$ given by Eq. \eqref{eq:realew-4.1} is not an EW for the following reason. By direct calculation $H$ has two negative eigenvalues and two positive eigenvalues, while any two-qubit EW has exact one negative eigenvalue and three positive eigenvalues \cite[Example 1]{ewspectral2008}. 

Second we show different cases of $W^+$ for an EW $W$ according to Corollary \ref{cor:W+}. Denote by $W(\theta):=\proj{\psi(\theta)}^\Gamma$ the EWs generated from $\ket{\psi(\theta)}$. Then the real part of $W(\theta)$ reads as 
\begin{equation*}
\bal
W(\theta)^+&=\frac{1}{2}(\proj{0,0}+\proj{1,1})\\
&+\frac{\cos\theta}{2}(\ketbra{0,1}{1,0}+\ketbra{1,0}{0,1}).
\eal
\end{equation*}
According to Lemma \ref{le:WandW+} (i), $W(\theta)^+$ is an EW if it is not positive semidefinite. By calculation the four eigenvalues of $W(\theta)^+$ are $\frac{1}{2},\frac{1}{2},\pm\frac{\cos\theta}{2}$.
It follows that $W(\theta)^+$ is an EW when $\cos\theta\neq 0$, and $W(\theta)^+$ is a separable state when $\cos\theta=0$. Moreover, it follows from Corollary \ref{cor:W+} (ii) that $W(\theta)$ detects no real and entangled state when $\cos\theta=0$. It means that there exists a CEW detecting no real and entangled state, even for the two-qubit system. However, it follows from Theorem \ref{le:realEW} (i) that every REW detects at least one real and entangled state. 
Hence, the family of EWs represented by $W(\theta)$ reflects the difference in detection power between REWs and CEWs.

Third, we propose a specific example showing that there exist complex and entangled states which cannot be detected by any REW. Based on Theorem \ref{le:realEW} (ii) it suffices to construct an entangled state whose real part is separable. Setting $\theta=\frac{\pi}{2}$, the entangled state $\rho:=\proj{\psi(\frac{\pi}{2})}=\frac{1}{2}(\ket{0,0}+i\ket{1,1})(\bra{0,0}-i\bra{1,1})$ cannot be detected by any REW according to Theorem \ref{le:realEW} (ii), as the real part $\rho^+=\frac{1}{2}(\proj{0,0}+\proj{1,1})$ is separable. Thus, $\rho$ is detected only by some CEWs, which clearly indicates the detection limitations of REWs.

\section{Detection power of REWs under local equivalences}
\label{sec:detpower}

To make up for the detection limitations of REWs, we further consider whether those entangled states detected by no REW can be detected by some EWs locally equivalent to the REWs. The local equivalence corresponds to local operations in the practical implementations. Thus, the motivation is to study how to operationally detect the entangled states of interest.
If an entangled state is detected by an REW assisted with local operations, we generally say this entangled state is detected by an REW up to the local equivalence. To make the problem above clear, we formulate it as the following conjecture.

\begin{conjecture}
    \label{cj:realew}
    (i) Every bipartite entangled state is detected by some $W\in\cE^{LU}$. 

    (ii) More generally, every bipartite entangled state is detected by some $W\in\cE^{SLOCC}$.
\end{conjecture}
 

By Definition \ref{def:ewsets}, it follows directly that Conjecture \ref{cj:realew} (ii) holds if Conjecture \ref{cj:realew} (i) holds, while the converse may not be true. Therefore, we mainly investigate Conjecture \ref{cj:realew} (i) under LU equivalence, and extend the results to Conjecture \ref{cj:realew} (ii) under SLOCC equivalence. In Lemma \ref{le:cj1-c1} we present several conditions to determine whether an entangled state can be detected by an EW in $\cE^{LU}$ or $\cE^{SLOCC}$. In Theorem \ref{thm:PPTEr4} we show the validity of Conjecture \ref{cj:realew} for a family of two-qutrit PPT entangled states of rank four. 
Finally, in Example \ref{ex:pptreal} we construct a $4\times 4$ PPT entangled state which cannot be detected by any REW but can be detected by some EW in $\cE^{LU}$.

By virtue of the PPT criterion, we may reduce the study of Conjecture \ref{cj:realew} to a subset of PPT states which are complex and entangled by the following lemma.

\begin{lemma}
\label{le:cj1-c1}
(i) Each NPT state is detected by an EW in $\cE^{LU}$.

(ii) A complex and PPT entangled state $\rho$ is detected by an EW in $\cE^{LU}$ ($\cE^{LOCC}$) if and only if there exists a state $\sigma$ LU (SLOCC) equivalent to $\rho$ such that the real part of $\sigma$, i.e. $\sigma^+$, is an entangled state.

(iii) A complex and entangled state $\rho$ is detected by an EW in $\cE^{LU}$ if and only if there exist two real orthogonal matrices $U_1,U_2$ such that $(U_1\ox U_2)^\dg\rho(U_1\ox U_2)$ is detected by an EW as $(D_1\ox D_2)W_r(D_1\ox D_2)^\dg$ for some diagonal unitaries $D_1,D_2$ and some REW $W_r$.

(iv) An entangled state can be detected by an EW in $\cE^{LU}$ if and only if it is detected by $(A\ox B)W(A\ox B)^\dg\in\cE^{SLOCC}$ for some REW $W$ and invertible matrices $A,B$, such that $(\abs{A}\ox\abs{B})W(\abs{A}\ox\abs{B})$ is real.  
\end{lemma}

\begin{proof}
(i) Suppose that $\rho$ is an NPT state. Let $\rho^\G=\sum_j p_j\proj{\psi_j}$ with $p_1<0$ be the spectral decomposition. It means that $\ket{\psi_1}$ is entangled, and thus $\proj{\psi_1}^\G$ is an EW. Then $\rho$ is detected by $\proj{\psi_1}^\G$ due to 
\beq 
\label{eq:lu-1}
\tr(\proj{\psi_1}^\G\rho)=\tr(\proj{\psi_1}\rho^\G)<0.
\eeq
By virtue of the Schmidt decomposition, we conclude that $\ket{\psi_1}$ is LU equivalent to a real and entangled state. It implies that $\ket{\psi_1}^\G\in\cE^{LU}$. From Eq. \eqref{eq:lu-1} we derive the assertion (i).

(ii) First, we show the ``If'' part. Assume that $\sigma\thicksim_{LU} \rho$ and $\sigma^+$ is entangled. It follows from Theorem \ref{le:realEW} (ii) that $\sigma$ is detected by some REW, namely $W_r$. According to the assumption there exists a product unitary $U\ox V$ such that $\sigma=(U\ox V)\rho (U\ox V)^\dg$. Then we obtain that
\beq
\label{eq:lu-2}
\bal
0>\tr(W_r\sigma)&=\tr(W_r(U\ox V)\rho (U\ox V)^\dg) \\
&=\tr((U\ox V)^\dg W_r(U\ox V)\rho).
\eal
\eeq
It implies that $\rho$ is detected by $(U\ox V)^\dg W_r(U\ox V)\in\cE^{LU}$.

Second, we show the ``Only if'' part. Assume that $\rho$ is detected by an EW as $(U\ox V)W_r(U\ox V)^\dg$ for some REW $W_r$ and local unitary $U\ox V$. Let $\sigma:=(U\ox V)^\dg\rho (U\ox V)$. Then we obtain that
\beq
\label{eq:lu-2.1}
\bal
0>&\tr((U\ox V)W_r(U\ox V)^\dg\rho) \\
=&\tr(W_r(U\ox V)^\dg\rho (U\ox V))=\tr(W_r\sigma).
\eal
\eeq
It implies that $\sigma$ is detected by the REW $W_r$. By Theorem \ref{le:realEW} (ii), $\sigma^+$ is both real and entangled.

Similar to the above discussion we can derive the same conclusion under SLOCC equivalence.

(iii) First, the ``If'' part follows directly from the fact:
\beq
\label{eq:lu-4.2}
\bal
0>&\tr\Big((D_1\ox D_2)W_r(D_1\ox D_2)^\dg(U_1\ox U_2)^\dg\rho(U_1\ox U_2)\Big) \\
=&\tr\Big(\big((U_1D_1)\ox(U_2D_2)\big)W_r\big((U_1D_1)^\dg\ox(U_2D_2)^\dg\big)\rho\Big).
\eal
\eeq
It means that $\rho$ is detected by an EW in $\cE^{LU}$.
Second, we show the ``Only if'' part. Suppose that $\rho$ is detected by $(X_1\ox X_2)W_r(X_1\ox X_2)^\dg$, where $W_r$ is an REW and $X_1\ox X_2$ is a local unitary. It follows from Lemma \ref{le:unitarydec} that $\forall i=1,2,~ X_i$ has the decomposition as $X_i=U_iD_iV_i$ for some real orthogonal $U_i,V_i$, and some diagonal unitary matrix $D_i$. By the decomposition it follows that 
\begin{widetext}
\beq
\label{eq:lu-4}
\bal
0>&\tr((X_1\ox X_2)W_r(X_1\ox X_2)^\dg\rho) \\
=&\tr\Big(\big((U_1D_1)\ox(U_2D_2)\big)\big[(V_1\ox V_2)W_r(V_1\ox V_2)^\dg\big]\big((U_1D_1)^\dg\ox(U_2D_2)^\dg\big)\rho\Big) \\
=&\tr\Big(\big[(V_1\ox V_2)W_r(V_1\ox V_2)^\dg\big]\big((U_1D_1)^\dg\ox(U_2D_2)^\dg\big)\rho\big((U_1D_1)\ox(U_2D_2)\big)\Big).
\eal
\eeq
\end{widetext}
Let $\sigma=\big((U_1D_1)^\dg\ox(U_2D_2)^\dg\big)\rho\big((U_1D_2)\ox(U_2D_2)\big)$. It follows from the last line in Eq. \eqref{eq:lu-4} that $\sigma$ is detected by an REW as $(V_1\ox V_2)W_r(V_1\ox V_2)^\dg$ for real orthogonal matrices $V_1$ and $V_2$. 
Denote $\tilde{W}_r=(V_1\ox V_2)W_r(V_1\ox V_2)^\dg$. It follows that $(U_1\ox U_2)^\dg\rho(U_1\ox U_2)$ is detected by $(D_1\ox D_2)\tilde{W}_r (D_1\ox D_2)^\dg$ by the last line in Eq. \eqref{eq:lu-4}. Thus, the ``Only if'' part holds.

(iv) The ``Only if'' part obviously holds as $\abs{U}=\sqrt{U^\dg U}=I$ for any unitary $U$. One can verify the ``Only if'' part by directly setting $A$ and $B$ to be unitary. Next, we show the ``If'' part. Suppose that $\rho$ is an entangled state detected by an EW as $(A\ox B)W_r(A\ox B)^\dg$ which is SLOCC equivalent to some REW $W_r$, where $A,B$ are two invertible matrices such that $(\abs{A}\ox\abs{B})W_r(\abs{A}\ox\abs{B})$ is real. It follows that $\tr((A\ox B)W_r(A\ox B)^\dg\rho)<0$. By the polar decomposition we obtain that $A=U\abs{A}$ and $B=V\abs{B}$ for some unitary $U,V$. Then we obtain that
\beq
\label{eq:lu-3}
\bal
0>&\tr((A\ox B)W_r(A\ox B)^\dg\rho) \\
=&\tr((U\ox V)(\abs{A}\ox\abs{B})W_r(\abs{A}\ox\abs{B})(U\ox V)^\dg\rho).
\eal
\eeq
Since $(\abs{A}\ox\abs{B})W_r(\abs{A}\ox\abs{B})$ is real, we conclude from Eq. \eqref{eq:lu-3} that $\rho$ is detected by an EW in $\cE^{LU}$.

This completes the proof.
\end{proof}

Here we make some necessary remarks about Lemma \ref{le:cj1-c1}. First, Lemma \ref{le:cj1-c1} (i) shows an evidence to support Conjecture \ref{cj:realew} (i), and by combining Lemma \ref{le:cj1-c1} (i) and (ii), to figure out Conjecture \ref{cj:realew} is precisely to study the complex and PPT entangled states whose real parts are separable. Second, Lemma \ref{le:cj1-c1} (ii) and (iii) present two necessary and sufficient conditions to verify the validity of Conjecture \ref{cj:realew}. Third, Lemma \ref{le:cj1-c1} (iv) establishes a connection between Conjecture \ref{cj:realew} (i) and (ii) by a necessary and sufficient condition. 

Now, we are able to provide a process of examining whether an entangled state can be detected by an EW in $\cE^{LU}$. This process is illustrated by a flowchart in Fig. \ref{fig:flow}. By following this flowchart, we verify Theorem \ref{thm:PPTEr4} and examine Example \ref{ex:pptreal}.

\begin{figure*}[htbp]
\centering
\includegraphics[width=0.8\textwidth]{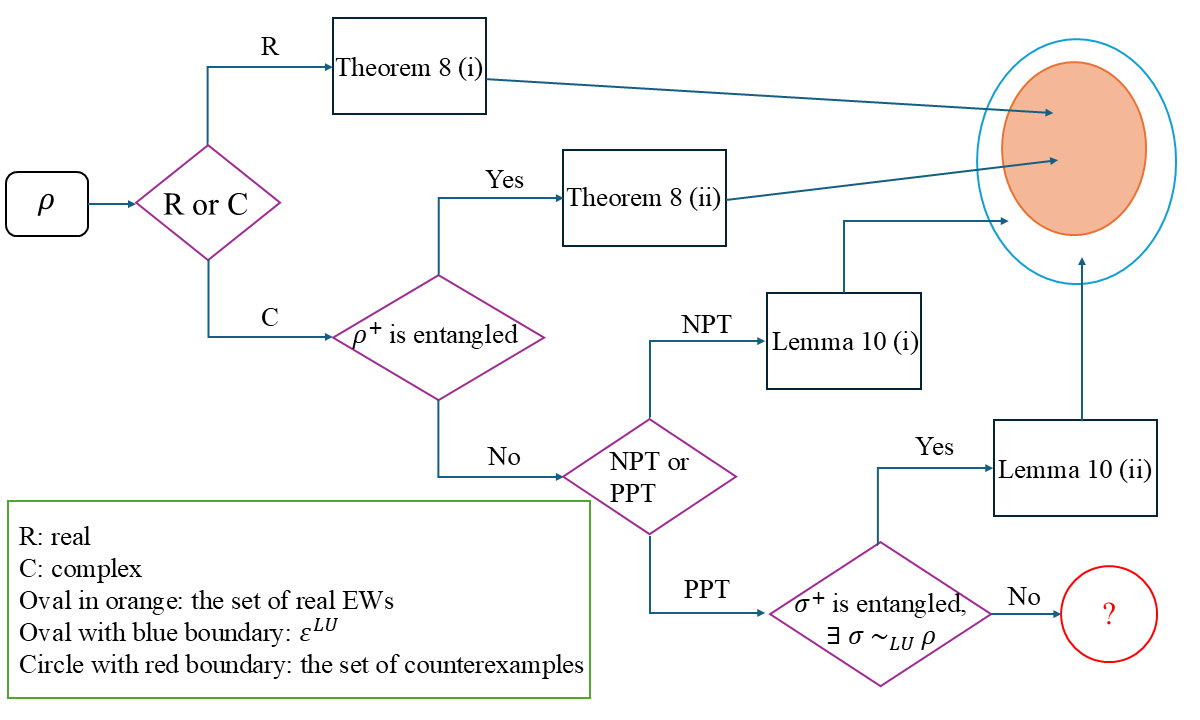}
\caption{This is the flowchart to examine whether an entangled state can be detected by an EW in $\cE^{LU}$. The set of REWs is certainly included in $\cE^{LU}$ by definition. Note that the existence of a counterexample against Conjecture \ref{cj:realew} is unknown yet. This is the reason why there is a ``?'' in the circle with red boundary.} 
\label{fig:flow}
\end{figure*}

By virtue of Lemma \ref{le:cj1-c1} (ii), we may switch the study on EWs to complex and PPT entangled states under local equivalences, which may bring some advantages because of the semidefinite positivity of states.
In the following we use the results derived above and follow the flowchart in Fig. \ref{fig:flow} to show the validity of Conjecture \ref{cj:realew} (i) for a family of two-qutrit PPT entangled states of rank four.

\begin{theorem}
\label{thm:PPTEr4}
(i) Every two-qutrit PPT entangled state of rank four generated by an unextendible product basis (UPB) can be detected by an EW in $\cE^{LU}$.

(ii) Every two-qutrit PPT entangled state of rank four is detected by an EW in $\cE^{SLOCC}$.
\end{theorem}

\begin{proof}
(i) There is an efficient scheme to produce PPT entangled states based on UPBs. Let $\ket{a_j,b_j}_{j=1}^k$ be a UPB. Then we may generate a PPT entangled state as
\beq
\label{eq:upbppte}
\rho:=\frac{1}{d-k}(I_d-\sum_{j=1}^k\proj{a_j,b_j}),
\eeq
where $d$ is the global dimension of the bipartite system. It is known from Ref. \cite{upbcmp2003} that every UPB in the two-qutrit system is LU equivalent to the quintuple $\ket{\psi_k}=\ket{\a_k}\ox\ket{\b_k}$, $k=0,\cdots,4$, where
\beq
\label{eq:3x3upb}
\bal
\ket{\a_0}&=\ket{0},\\
\ket{\a_1}&=\ket{1},\\
\ket{\a_2}&=\cos\t_A\ket{0}+\sin\t_A\ket{2},\\
\ket{\a_3}&=\sin\g_A(\sin\t_A\ket{0}-\cos\t_A\ket{2})+\cos\g_A e^{i\phi_A}\ket{1},\\
\ket{\a_4}&=\frac{1}{N_A}(\sin\g_A\cos\t_A e^{i\phi_A}\ket{1}+\cos\g_A\ket{2}),\\
\ket{\b_0}&=\ket{1},\\
\ket{\b_1}&=\sin\g_B(\sin\t_B\ket{0}-\cos\t_B\ket{2})+\cos\g_B e^{i\phi_B}\ket{1},\\
\ket{\b_2}&=\ket{0},\\
\ket{\b_3}&=\cos\t_B\ket{0}+\sin\t_B\ket{2},\\
\ket{\b_4}&=\frac{1}{N_B}(\sin\g_B\cos\t_B e^{i\phi_B}\ket{1}+\cos\g_B\ket{2}).
\eal
\eeq
Here the six real parameters are the angles, denoted by $\g_{A(B)},\t_{A(B)},\phi_{A(B)}$, and the normalization constants are given by the formula:
\beq
\label{eq:3x3upb-1}
N_{A(B)}=\sqrt{\cos^2\g_{A(B)}+\sin^2\g_{A(B)}\cos^2\t_{A(B)}}.
\eeq
Due to this essential LU equivalence and Eq. \eqref{eq:upbppte}, it suffices to consider if the PPT entangled states generated by the UPBs formulated above can be detected by an EW in $\cE^{LU}$. 
In view of this, we first denote a projector onto the subspace spanned by the UPB formulated by Eqs. \eqref{eq:3x3upb} and \eqref{eq:3x3upb-1} as follows:
\beq
\label{eq:3x3upb-2}
P_{\g_{A(B)},\t_{A(B)},\phi_{A(B)}}=\sum_{j=0}^4\proj{\a_k,\b_k}
\eeq
w.r.t. the six real parameters $\g_{A(B)},\t_{A(B)},\phi_{A(B)}$. Let $D_A=\diag(1,e^{i(-\phi_A)},1)$ and $D_B=\diag(1,e^{i(-\phi_B)},1)$. By observing Eq. \eqref{eq:3x3upb-2}, we obtain that 
$$(D_A\ox D_B)P_{\g_{A(B)},\t_{A(B)},\phi_{A(B)}}(D_A\ox D_B)^\dg$$ 
is a real projector. It follows that 
\beq
\label{eq:3x3upb-3}
\sigma:=\frac{1}{4} \left(I_9-(D_A\ox D_B)P_{\g_{A(B)},\t_{A(B)},\phi_{A(B)}}(D_A\ox D_B)^\dg \right)
\eeq
is a real and PPT entangled state. From Theorem \ref{le:realEW} (i) we conclude that $\sigma$ is detected by an REW, and thus the PPT entangled states generated by the UPBs formulated above is detected by an EW in $\cE^{LU}$. Therefore, assertion (i) holds, as every UPB in the two-qutrit system is LU equivalent to the UPBs formulated above.

(ii) According to Ref. \cite[Proposition 4.]{pptrank4jmp}, any two-qutrit PPT entangled state of rank four is SLOCC equivalent to a real state that is invariant under partial transposition. Specifically, for any $3\times 3$ PPT entangled state $\rho$, there exist invertible $A,B\in\cM_3(\bbC)$ such that $\sigma=(A\ox B)\rho(A\ox B)^\dg$ have the property $\sigma^\Gamma=\sigma$. Furthermore, the PPT entangled state $\sigma$ is real and can be formulated as $\sigma=C^\dg C$, where $C=[C_0,C_1,C_2]$ and
\beq
\label{eq:rank4ppte}
C_0=
\bma
0 & a & b \\
0 & 0 & 1 \\
0 & 0 & 0 \\
0 & 0 & 0 
\ema,~
C_1=
\bma
0 & 0 & 0 \\
0 & 0 & c \\
0 & 0 & 1 \\
1 & 0 & -\frac{1}{d}
\ema,~
C_2
\bma
0 & -\frac{1}{b} & 0 \\
0 & 1 & 0 \\
1 & -c & 0 \\
d & 0 & 0
\ema,
\eeq
with $a,b,c,d>0$.
Since $\sigma$ given by Eq. \eqref{eq:rank4ppte} is real-formula, it is detected by an REW from Theorem \ref{le:realEW} (i). Due to the essential SLOCC equivalence derived in Ref. \cite{pptrank4jmp}, we conclude that the assertion (ii) holds.

This completes the proof.
\end{proof}

It is known that there is no PPT entangled state in the two-qubit and qubit-qutrit systems. Hence, it follows from Lemma \ref{le:cj1-c1} (i) that Conjecture \ref{cj:realew} holds for the bipartite systems with a global dimension less than $6$. Theorem \ref{thm:PPTEr4} shows that Conjecture \ref{cj:realew} holds for a type of two-qutrit states of rank four.
We further extend the study to the bipartite system associated with $\bbC^4\ox\bbC^4$. We construct a concrete example as below to show the existence of the two-quqart PPT entangled states which cannot be detected by any REW but can be detected by some $W\in\cE^{LU}$. This example also supports Conjecture \ref{cj:realew} (i). Note that the following example stems from the PPT entangled state proposed in Ref. \cite{pptsqconjlin2019} which is related to the well-known PPT square conjecture.

\begin{example}
\label{ex:pptreal}
We construct a non-normalized state supported on $\bbC^4\ox\bbC^4$ as
\beq
\label{eq:pptrealex-1}
\bal
\rho=&(\ket{00}+\ket{11}+\ket{22})(\bra{00}+\bra{11}+\bra{22}) \\
+&\proj{02}+\proj{20}+\proj{12}+\proj{21}\\
+&(\ket{01}+\ket{10}+\ket{33})(\bra{01}+\bra{10}+\bra{33}) \\
+&\proj{03}+\proj{30}+\proj{13}+\proj{31}.
\eal
\eeq
The state $\rho$ has a real density matrix.
One can verify that $\rho^\G$ is also positive semidefinite and $\cR(\rho)$ is not spanned by product vectors. Thus, $\rho$ given by Eq. \eqref{eq:pptrealex-1} is a real and PPT entangled state which can be detected by an REW according to Theorem \ref{le:realEW} (i). Next, we propose a state LU equivalent to $\rho$ as below:
\beq
\label{eq:pptrealex-2}
\sigma:=(I_4\ox\diag(1,1,i,i))\rho(I_4\ox\diag(1,1,-i,-i)).
\eeq
Due to the LU equivalence, $\sigma$ is also a PPT entangled state but its density matrix is not real any more. One can verify that the real part of $\sigma$ is separable as follows:
\beq
\label{eq:pptrealex-3}
\bal
\sigma^+ &=(\ket{00}+\ket{11})(\bra{00}+\bra{11})+\proj{22}\\
&+(\ket{01}+\ket{10})(\bra{01}+\bra{10})+\proj{33}\\
&+\proj{02}+\proj{20}+\proj{12}+\proj{21}\\
&+\proj{03}+\proj{30}+\proj{13}+\proj{31}\\
&=\frac{1}{2}(\ket{0}+\ket{1})(\bra{0}+\bra{1})\ox(\ket{0}+\ket{1})(\bra{0}+\bra{1})\\
&+\frac{1}{2}(\ket{0}-\ket{1})(\bra{0}-\bra{1})\ox(\ket{0}-\ket{1})(\bra{0}-\bra{1})\\
&+\proj{22}+\proj{02}+\proj{20}+\proj{12}+\proj{21}\\
&+\proj{33}+\proj{03}+\proj{30}+\proj{13}+\proj{31}.
\eal
\eeq 
Therefore, $\sigma$ formulated by Eq. \eqref{eq:pptrealex-2} is a complex and PPT entangled state whose real part is separable. It is known from Theorem \ref{le:realEW} (ii) that $\sigma$ cannot be detected by any REW. Nevertheless, since $\sigma$ is LU equivalent to a real and entangled state $\rho$, we conclude that $\sigma$ is detected by an EW in $\cE^{LU}$. This example, as the high-dimensional exploration, supports Conjecture \ref{cj:realew} (i) in the case when a PPT entangled state has a separable real part.
\end{example}

According to the results derived above, we believe that Conjecture \ref{cj:realew} holds at least for the bipartite systems with low local dimensions.
The remaining part of Conjecture \ref{cj:realew} (i) is to show that the complex and PPT entangled states whose real parts are separable can be detected by an EW in $\cE^{LU}$. Therefore, to attack Conjecture \ref{cj:realew}, it is necessary to link the states with high local dimensions to those with low local dimensions, and characterize the set of PPT entangled states whose real parts are separable.

\section{Characterization of the real parts of PPT states under LU equivalence}
\label{sec:ppt+sepr}

Recall that by Lemma \ref{le:cj1-c1} (i) and (ii), Conjecture \ref{cj:realew} has been reduced to studying complex and PPT entangled states whose real parts are separable, and the problem of verifying Conjecture \ref{cj:realew} can be equivlently transformed to testing whether the real parts of the states mentioned above preserve separable under local equivalences. For this purpose, in this section we further characterize the real parts of the complex and PPT entangled states under LU equivalence. In Lemma \ref{le:realpartspace}, we show some relations between a complex PPT state itself and its real part. In Lemma \ref{le:prs-prop}, we characterize some properties of the special set $\cP_{rs}(m,n)$ of PPT states, given in Definition \ref{def:pptseprset}. The properties could be used to examine whether a counterexample against Conjecture \ref{cj:realew} exists.

For a complex PPT state, itself is related to its real part. One can infer some properties of a bipartite state $\rho_{AB}$ from its real part $\rho_{AB}^+$ as follows.
\begin{lemma}
\label{le:realpartspace}
(i) Suppose that $\rho_{AB}$ is a bipartite state supported on $\bbC^m \ox \bbC^n$. If $\rank(\rho^+_A)=p$ and $\rank(\rho^+_B)=q$, where $\rho^+_A$ and $\rho^+_B$ are two reduced states of $\rho_{AB}^+$, then $\rho_{AB}$ is indeed a state supported on $\bbC^p\ox\bbC^q$.


(ii) Denote by $\rho_{AB}=\sum_{x,y=1}^m \ketbra{x}{y}\ox \rho_{xy}$ a PPT state supported on $\bbC^m \ox \bbC^m$. Suppose that the real part $\rho_{AB}^+$ is diagonal, namely $\rho_{AB}^+=\sum_{j,k=1}^{m}p_{j,k}\proj{j,k}$, $\forall p_{jk}\geq 0$. Then the following assertions hold.

(ii.a) If $p_{j,k}=0$, then the entries in the $k$-th rows and $k$-th columns of $\rho_{j1},\rho_{j2},\cdots,\rho_{jm}$ and $\rho_{1j},\rho_{2j},\cdots,\rho_{mj}$ are all zero. 

(ii.b) If there exist distinct $k_1,\cdots,k_s\in\{1,\cdots,m\}$ and distinct $l_1,\cdots,l_t\in\{1,\cdots,m\}$ such that 
\begin{eqnarray*}
&&p_{x,k_1}=p_{x,k_2}=\cdots=p_{x,k_s}=0, \\
&&p_{y,l_1}=p_{y,l_2}=\cdots=p_{y,l_t}=0,
\end{eqnarray*}
both hold for two distinct $x,y\in\{1,\cdots,m\}$, then for any $k\in\{k_1,\cdots,k_s\}\cup\{l_1,\cdots,l_t\}$, both the $k$-th row and the $k$-th column of the non-diagonal blocks $\rho_{xy}$ and $\rho_{yx}$ of $\rho_{AB}$ are filled with zeros.

(ii.c) If $\rho_{AB}^+=\sum_{j=1}^m p_{j,j}\proj{\pi_1(j),\pi_2(j)}$ for any two permutations $\pi_1,\pi_2$ of $\{1,2,\cdots,m\}$, then $\rho_{AB}$ must be the same as $\rho_{AB}^+$ which is separable.
\end{lemma}

We put the detailed proof of Lemma \ref{le:realpartspace} in Appendix \ref{sec:proof-2}. It follows from Lemma \ref{le:realpartspace} (i) that a bipartite state has the same support as its real part, which is useful to specify the supporting space of a bipartite state from its real part. Lemma \ref{le:realpartspace} (ii) applies to the PPT states with diagonal real parts, by which one can roughly know what the PPT state is from the diagonal entries of its real part.

Another way of studying Conjecture \ref{cj:realew} is to show the existence or non-existence of the counterexample to Conjecture \ref{cj:realew} by deeply characterizing the properties that a counterexample should satisfy. Assume that there is a counterexample, namely $\rho_{AB}$, which has to be PPT by Lemma \ref{le:cj1-c1} (i). Further, by Lemma \ref{le:cj1-c1} (ii), we equivalently derive that the real part of $\rho_{AB}$ preserves separable under LU equivalence, i.e., $\left((U\ox V)\rho_{AB}(U\ox V)^\dg\right)^+$ is always separable for any product unitary $U\ox V$. It motivates us to investigate Conjecture \ref{cj:realew} (i) from a set-theoretic perspective, by focusing on the set of PPT states namely $\cP_{rs}(m,n)$ defined in Definition \ref{def:pptseprset}. Based on the analysis above, the entangled states in $\cP_{rs}(m,n)$ contradicts Conjecture \ref{cj:realew} (i). Recall the statement below Definition \ref{def:pptseprset} that the set of all separable states $\cS(m,n)$ is included in $\cP_{rs}(m,n)$. It implies that proving Conjecture \ref{cj:realew} (i) is equivalent to showing $\cP_{rs}(m,n)=\cS(m,n)$. Hence, it is necessary to characterize $\cP_{rs}(m,n)$, and compare with $\cS(m,n)$. Here, we show some properties shared by such two sets as below.

\begin{lemma}
    \label{le:prs-prop}
(i) The set $\cP_{rs}(m,n),~\forall m,n\geq 2$, given in Definition \ref{def:pptseprset} is convex and closed. \\
(ii) If $\rho_{AB}\in\cP_{rs}(m,n)$, then 

(a) $\rho_{AB}^T=\rho_{AB}^*\in\cP_{rs}(m,n)$; 

(b) $\tilde{\rho}_{AB}\in\cP_{rs}(m,n)$, for any $\tilde{\rho}_{AB}\sim_{LU}\rho_{AB}$; 

(c) $\rho_{AB}^\Gamma\in\cP_{rs}(m,n)$;

(d) $\rho_{BA}\in\cP_{rs}(n,m)$. 
\end{lemma}

We put the detailed proof of Lemma \ref{le:prs-prop} into Appendix \ref{sec:proof-2}. One can verify that the set $\cS(m,n)$ also shares the properties given in Lemma \ref{le:prs-prop}. This evidence does not violate $\cP_{rs}(m,n)=\cS(m,n)$ which is equivalent to Conjecture \ref{cj:realew} (i). Furthermore, by generalizing the LU equivalence in Definition \ref{def:pptseprset} to SLOCC equivalence, one can analogously show that the generalized set of $\cP_{rs}(m,n)$ also has similar properties as Lemma \ref{le:prs-prop}.

\section{Projecting EWs by local transformations}
\label{subsec:cj2}

It is known that decomposable EWs detect NPT states only, and PPT entangled states can be detected only by non-decomposable EWs. More generally, in Lemma \ref{le:ews-ge} we present the essential relations between EWs and the generalized concepts, namely decomposable matrices and block-positive matrices. Moreover, Theorem \ref{thm:PPTEr4} provides an evidence supporting Conjecture \ref{cj:realew} for two-qutrit PPT entangled states, which suggests that Conjecture \ref{cj:realew} may hold for the entangled states supported on low-dimensional spaces. To move on, we discuss in Lemma \ref{le:cj2-c1} how to locally project an EW to another supported on a lower-dimensional space.



We shall introduce two generalized concepts closely related to EWs. First, a bipartite matrix $W$ is called \emph{decomposable} if it can be written as $W=X^\Gamma+Y$, where $X,Y$ are both positive semidefinite. Second, a bipartite matrix $W$ is called \emph{block-positive} if it is of the form $W:=(id_m\ox\Phi)(X)$ for some positive semidefinite $X\in\cM_m(\bbC)\ox\cM_m(\bbC)$ and some positive map $\Phi:\cM_m(\bbC)\to\cM_n(\bbC)$. We shall elucidate EWs with relation to the two generalized concepts as below.

\begin{lemma}
\label{le:ews-ge}
(i) Every decomposable matrix is block-positive. The converse of this statement is true if and only if $m\times n\leq 6$.

(ii) Every EW is block-positive, and a block-positive matrix is an EW if and only if it is not positive semidefinite.

(iii) Every decomposable matrix is an EW if and only if it is not positive semidefinite. 

(iv) An EW is decomposable if and only if it detects NPT states only. An EW is non-decomposable if and only if it detects both NPT and PPT entangled states. 
\end{lemma}

The assertions are collected from references. For example, one may refer to Ref. \cite{inverseew2018} for more details. It is worth mentioning that decomposable EWs detect NPT states only, and any non-decomposable EW detects both NPT and PPT entangled states from Lemma \ref{le:ews-ge} (iv). It refines the commonly used statement in literatures that PPT entangled states are detected only by non-decomposable EWs, which does not show the fact clearly that non-decomposable EWs also detect NPT states. One can show that a non-decomposable EW must detect an NPT state as follows. Suppose that $W_p$ is a non-decomposable EW that detects a PPT entangled state $\rho$. Let $\rho=\sum_j\proj{\psi_j}$. Due to $\tr(W_p\rho)<0$, there is one $\ket{\psi_j}$ such that $\bra{\psi_j}W_p\ket{\psi_j}<0$. It implies that $\ket{\psi_j}$ is an entangled state detected by $W_p$. Thus, $W_p$ also detects an NPT state $\ket{\psi_j}$. We also illustrate the relations given in Lemma \ref{le:ews-ge} by Fig. \ref{fig:ews-ge}.

\begin{figure}[htbp]
\centering
\includegraphics[width=0.45\textwidth]{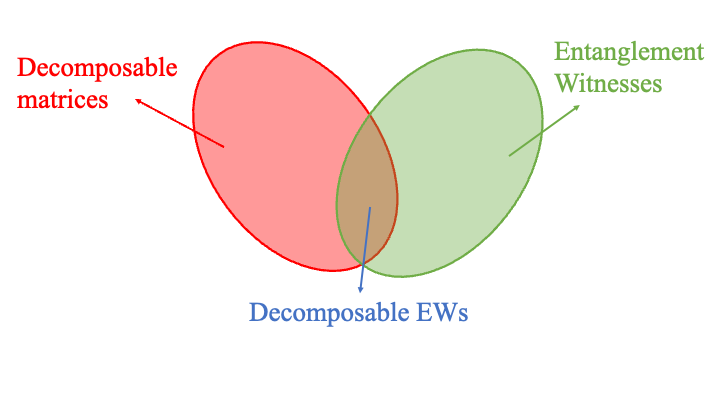}
\caption{Classification of block-positive matrices: The union of the two ovals represents the set of block-positive matrices. The oval in red represents the set of decomposable matrices. The oval in green represents the set of EWs. The intersection of the two ovals represents the set of decomposable EWs.} 
\label{fig:ews-ge}
\end{figure}



Next, we study the local projections of EWs. Note that the local projection of an EW may not be an EW any more. The projection of an EW by applying a product operator is still an EW, if and only if the projection is non-positive semidefinite, as $X\ket{a,b}$ remains a product vector for any product operator $X$ and any product vector $\ket{a,b}$. As it is known that all NPT states supported on $\bbC^2\ox\bbC^n$ are distillable \cite{2ndis2000}, to establish a potential connection with entanglement distillation, we study whether an $m\times n$ EW can be projected to a $2\times n$ EW by local transformations as follows. 

\begin{lemma}
\label{le:cj2-c1}
(i) If a bipartite EW can be locally projected to a $2\times n$ EW, then it can be further locally projected to a two-qubit EW.

(ii) Suppose that $\rho_{AB}$ is a $1$-undistillable NPT state. Then $\rho_{AB}^\Gamma$ cannot be locally projected to a $2\times n$ EW.

(iii) If an EW can be locally projected to a non-decomposable EW detecting some PPT entangled state, then itself must be non-decomposable.
\end{lemma}

We put the detailed proof of Lemma \ref{le:cj2-c1} into Appendix \ref{sec:proof-3}.
In the proof of Lemma \ref{le:cj2-c1} (i), we show how to project a $2\times n$ pure entangled state to a two-qubit NPT state by a local projector.
In the proof of Lemma \ref{le:cj2-c1} (ii), we prove that such $m\times n$ EWs constructed by the partial transposition of $1$-undistillable NPT states cannot be locally projected to $2\times n$ EWs. 
From Lemma \ref{le:cj2-c1} (iii) we conclude that only non-decomposable EWs are possible to be locally projected to $2\times 4$ EWs which detect PPT entangled states. 

Finally, similar to the approaches given in the proof of Lemma \ref{le:cj2-c1} (i), we propose a method for estimating how small the local dimensions of the spaces supporting the projected NPT states and projected EWs are. 


\begin{lemma}
\label{le:projmindim}
(i) Suppose that $\rho$ is an NPT state. Let $\ket{y}\in\mathbb{C}^{m(y)}\ox\mathbb{C}^{n(y)}$ be in the negative eigenspace of $\rho^\G$. Denote by $p=\min\limits_{\ket{y}}\{m(y),n(y)\}$. Then $\rho$ can be locally projected to an NPT state supported on $\mathbb{C}^{p}\ox\mathbb{C}^{p}$.

(ii) Suppose that $W$ is an EW. Let $\ket{y}\in\mathbb{C}^{m(y)}\ox\mathbb{C}^{n(y)}$ be in the negative eigenspace of $W$. Denote by $p=\min\limits_{\ket{y}}\{m(y),n(y)\}$. Then $W$ can be locally projected to an EW supported on $\mathbb{C}^{p}\ox\mathbb{C}^{p}$.
\end{lemma}

One may refer to Appendix \ref{sec:proof-3} for the detailed proof of Lemma \ref{le:projmindim}. Note that the smaller local dimensions of the supporting space derived in Lemma \ref{le:projmindim} may not be the minimum ones.

\section{Conclusions}
\label{sec:con}

In this paper we revealed the differences in detection power between REWs and CEWs, and analyzed the detection power of REWs under local equivalences, especially LU equivalence. First, we have explained that an REW must detect a real and entangled state, and conversely a real and entangled state must be detected by some REW. We also presented a necessary and sufficient condition for the entangled states detected by REWs. This condition completely relies on the separability of their real parts. With this condition we confirmed the existence of the entangled states which cannot be detected by any REW. It clearly indicates the detection limitations by using REWs only. For this fact, we further investigated which entangled states can be detected by the EWs locally equivalent to some REWs, i.e. Conjecture \ref{cj:realew}. We proved Conjecture \ref{cj:realew} for all NPT states, and presented a necessary and sufficient condition for the complex and PPT entangled states that can be detected by the EWs in $\cE^{LU}$. We further proved Conjecture \ref{cj:realew} for a family of two-qutrit PPT entangled states, and constructed a $4\times 4$ PPT entangled state of a complex density matrix satisfying the criterion above. Another way to attack Conjecture \ref{cj:realew} is to examine the existence of a counterexample. We proposed an equivalent method to examine the existence from a set-theoretic perspective. Using this method, the results do not suggest the existence of a counterexample. Based on the supporting evidence, we believe that Conjecture \ref{cj:realew} may hold at least for low-dimensional states. To lower the local dimensions of the systems, we finally investigated if it is possible to locally project an EW to another supported on the lower-dimensional space.

The main target of future work is to generally prove Conjecture \ref{cj:realew} as the local dimensions increase, or to numerically test the counterexamples by proposing conditions that are easy to verify. As byproducts, there are several interesting directions to move on. First, we may build more essential connections between Conjecture \ref{cj:realew} (i) and (ii). Second, we may study more deeply the inclusion relation between the two sets $\cP_{rs}(m,n)$ and $\cS(m,n)$, because such relations are connected to the validity of Conjecture \ref{cj:realew} by our results. Third, to link entangled states and EWs to those supported on lower dimensional spaces, we may further investigate their local projections.
Finally, the connections and differences in terms of physical implications between REWs and CEWs are not evident yet, while the physical implications are instructive to characterize and detect entanglement.

\section*{acknowledgements}

The authors appreciate the anonymous reviewers for helping improve the quality of this paper. Y.S. is funded by the NNSF of China (Grant No. 12401597), the Basic Research Program of Jiangsu (Grant No. BK20241603), and the Wuxi Science and Technology Development Fund Project (Grant No. K20231008). L.C. is supported by the NNSF of China (Grant No. 12471427).
Z.B. was supported by the NNSF of China (Grant No. 12104186).

\bibliography{witness}


\appendix

\section{Proofs of some results in Sec. \ref{sec:limit}}
\label{sec:proof-1}

\textbf{Proof of Lemma \ref{le:ew_rho}.}
(i) If $W_1,W_2\in\cE_{\rho}$ which means that $\tr(W_1\rho)<0$ and $\tr(W_2\rho)<0$, then we obtain for any $t\in[0,1]$,
\beq 
\label{eq:ew_rho-1}
\tr\big((tW_1+(1-t)W_2)\rho\big)=t\tr(W_1\rho)+(1-t)\tr(W_2\rho)<0.
\eeq
It means that $tW_1+(1-t)W_2\in\cE_{\rho}$.

(ii) This fact follows from the equality below:
\beq
\label{eq:eq_rho-2}
\tr(W\rho)=\tr(\rho^T W^T)=\tr(\rho^* W^*).
\eeq
If $W$ is an EW which detects the entangled state $\rho$, it follows from Eq. \eqref{eq:eq_rho-2} that $\tr(W\rho)=\tr(\rho^* W^*)<0$. By the definition of EW, we also conclude that, for any product state $\ket{a,b}$,
\beq 
\label{eq:w*-prod}
\bra{a,b}W^*\ket{a,b}=\bra{a^*,b^*} W\ket{a^*,b^*}\geq 0.
\eeq
It follows from Eqs. \eqref{eq:eq_rho-2} and \eqref{eq:w*-prod} that $\rho^*$ is entangled if $\rho$ is entangled, and  $W^*$ is an EW detecting $\rho^*$ if and only if $W$ is an EW detecting $\rho$. Furthermore, if $\cE_\rho$ contains an REW $W$, then $W$ is also included in $\cE_{\rho^*}$ by the above assertion, i.e. $W\in\cE_\rho\cap\cE_{\rho^*}$. Next, if there exists an EW $W\in\cE_\rho\cap\cE_{\rho^*}$, it follows that $\tr(W\rho)<0$ and $\tr(W\rho^*)<0$. Based on $\tr(W\rho^*)=\tr(\rho W^*)=\tr(W^*\rho)<0$, we conclude that $\tr(W^+\rho)=\frac{1}{2}(\tr(W\rho)+\tr(W^*\rho))<0$. Due to Eq. \eqref{eq:w*-prod} we also conclude that $\bra{a,b}W^+\ket{a,b}\geq 0$, for any product state $\ket{a,b}$. It follows that $W^+$ is an REW which detects the entangled state $\rho$, i.e $W^+\in\cE_\rho$.

(iii) According to Lemma \ref{le:ewpt}, for an EW $W$, $W^\G$ is an NPT state if $W^\G$ is positive semidefinite, and otherwsie $W^\G$ is an EW. Based on this fact, this assertion follows from the equality below:
\beq
\label{eq:eq_rho-3}
\tr(W^\G\rho^\G)=\tr(W(\rho^\G)^\G)=\tr(W\rho).
\eeq
Since $\rho$ is PPT entangled, one can verify that $\rho^\Gamma$ is also a PPT entangled state. If $W\in\cE_\rho$, we equivalently obtain that $\tr(W^\G\rho^\G)=\tr(W\rho)<0$, from Eq. \eqref{eq:eq_rho-3}. Then $W^\Gamma$ is non-positive semidefinite and has to be an EW by Lemma \ref{le:ewpt}. Thus, $W\in\cE_\rho$ is equivalent to $W^\Gamma\in\cE_{\rho^\Gamma}$.

(iv) This fact follows from the equality below:
\beq
\label{eq:eq_rho-5}
\tr(W\rho)+\tr(W\rho^*)=2\tr(W\rho^+).
\eeq
Since $\rho^+$ is separable, we conclude that $\tr(W\rho^+)\geq 0$ for any EW $W$, which means $\tr(W\rho)+\tr(W\rho^*)\geq 0$. 

(v) This fact follows from the equality below:
\beq
\label{eq:eq_rho-4}
\bal
&\tr((U\ox V)W(U\ox V)^\dg\sigma) \\
=&\tr((U\ox V)W(U\ox V)^\dg(U\ox V)\rho(U\ox V)^\dg) \\
=&\tr(W\rho).
\eal
\eeq
The equality \eqref{eq:eq_rho-4} shows that $\rho$ is detected by an EW $W$ if and only if $\sigma:=(U\ox V)\rho(U\ox V)^\dg$ is detected by the EW as $(U\ox V)W(U\ox V)^\dg$.

This completes the proof.
\qed

\textbf{Proof of Lemma \ref{le:WandW+}.}
(i) Since $W$ is an EW, for any $t\in[0,1]$ and any product state $\ket{a,b}$, it follows that
\beq
\bal
\label{eq:W_t-1}
&\bra{a,b}\left(tW+(1-t)W^*\right)\ket{a,b} \\
=&t\bra{a,b}W\ket{a,b}+(1-t)\bra{a,b}W^*\ket{a,b} \\
=&t\bra{a,b}W\ket{a,b}+(1-t)\bra{a^*,b^*}W\ket{a^*,b^*}\geq 0.
\eal
\eeq
Thus, by definition we conclude that $W_t$ is an EW if and only if it is non-positive semidefinite. 

(ii) First, we prove the ``If'' part. Suppose that $W$ detects an entangled state $\rho_t:=t\rho+(1-t)\rho^*$ for some state $\rho$. It follows from $\tr(W\rho^*)=\tr(W^*\rho)$ that
\beq
\label{eq:W_t-2}
0>\tr(W\rho_t)=t\tr(W\rho)+(1-t)\tr(W\rho^*)
=\tr(W_t\rho).
\eeq
Then we conclude from assertion (i) that $W_t$ is an EW, and $\rho$ is an entangled state detected by $W_t$.
Second, we prove the ``Only if'' part. Suppose that $W_t$ detects an entangled state $\rho$. It implies that $\tr(W_t\rho)<0$. According to Eq. \eqref{eq:W_t-2}, we conclude that $\tr(W\rho_t)=\tr(W_t\rho)<0$, and thus $W$ detects an entangled state as $\rho_t:=t\rho+(1-t)\rho^*$, where $\rho$ is detected by $W_t$.

This completes the proof.
\qed

\textbf{Proof of Corollary \ref{cor:W+}.}
(i) This assertion follows directly from Lemma \ref{le:WandW+} (ii), by assigning the coefficient $t$ of $W_t$ as $\frac{1}{2}$.

(ii) According to Lemma \ref{le:WandW+}, $W^+$ is a state if and only if $W$ detects no real and entangled state. Next, we consider when $W^+$ becomes a PPT state. It suffices to verify the semidefinite positivity of $(W^+)^\Gamma$. By the equality $(W^+)^\Gamma=(W^\Gamma)^+$ from Eq. \eqref{eq:treq-2}, we may take $(W^+)^\Gamma$ as the real part of $W^\Gamma$. It follows from Lemma \ref{le:ewpt} that $W^\Gamma$ is either an NPT state or an EW. First, when $W^\Gamma$ is an NPT state, the real part of $W^\Gamma$, i.e. $(W^+)^\Gamma$ is positive semidefinite, and thus $W^+$ is a PPT state. Second, when $W^\Gamma$ is an EW, it follows from Lemma \ref{le:WandW+} and assertion (i) that $(W^\Gamma)^+$ is positive semidefinite if and only if $W^\Gamma$ detects no real and entangled state. To sum up, we obtain the assertion (ii).

(iii) Under the precondition that $W^+$ is a state, $W^+$ is either a PPT one or an NPT one. Thus, by virtue of the assertion (ii), we obtain the assertion (iii). For the last statement, it follows from Lemma \ref{le:ew_rho} (iii) that $\rho$ is detected by $W^\Gamma$ if and only if $\rho^\Gamma$ is detected by $W$, for any PPT entangled $\rho$. However, this condition contradicts that $W$ detects no real and entangled state. Therefore, the real and entangled states detected by $W^\Gamma$ must be NPT.

This completes the proof.
\qed

\section{Proofs of some results in Sec. \ref{sec:ppt+sepr}}
\label{sec:proof-2}

\textbf{Proof of Lemma \ref{le:realpartspace}.}
(i) Since $\rho_A^+$ is of rank $p$, there is a unitary $U$ such that $U\rho_A^+ U^\dg=\diag(a_1,\cdots,a_p,0,\cdots,0)$. Similarly, there is a unitary $V$ such that $V\rho_B^+ V^\dg=\diag(b_1,\cdots,b_q,0,\cdots,0)$, as $\rho_B^+$ is of rank $q$. Thus, up to LU equivalence, $\rho_{AB}^+$ is indeed supported on $\bbC^p\ox\bbC^q$, and can be assumed as a block matrix: 
\begin{equation}
    \label{eq:rp-1}
\rho_{AB}^+ = 
\begin{bmatrix}
    M_{11} & M_{12} & \cdots & M_{1p} & O \\
    M_{12}^\dg & M_{22} & \cdots & M_{2p} & O \\
    \vdots & \vdots & \cdots & \vdots & \vdots \\
    M_{1p}^\dg & M_{2p}^\dg & \cdots & M_{pp} & O \\
    O & O & \cdots & O & O 
\end{bmatrix},
\end{equation}
where each $O$ is a zero matrix with proper rows and columns, and each $M_{ij}$ is an $n\times n$ matrix in the form as
\begin{equation}
    \label{eq:rp-1.1}
M_{ij} = 
\begin{bmatrix}
   M_{ij}' & O_{q\times (n-q)} \\
    O_{(n-q)\times q} & O_{(n-q)\times (n-q)}
\end{bmatrix}.
\end{equation}
Since $\rho_{AB}^-$ is real and skew-symmetric, it implies that $(\rho_{AB}^-)^T=-\rho_{AB}^-$ and all diagonal entries are zero. One can verify that for any positive semidefinite matrix, if the diagonal entry $(i,i)$ is zero, then all entries in the $i$-th row and the $i$-th column have to be zero. Thus, to ensure that $\rho_{AB}=\rho_{AB}^+ + i \rho_{AB}^-$ is positive semidefinite, we conclude that $\rho_{AB}$ is also supported on $\bbC^p\ox\bbC^q$ based on the two forms given by Eqs. \eqref{eq:rp-1} and \eqref{eq:rp-1.1}.

(ii) We first show the assertion (ii.a). Since $\rho_{AB}^-$ is real and skew-symmetric, which implies that the diagonal entries of $\rho_{AB}^-$ are all zero, then the $(k,k)$ entry of $\rho_{jj}$ is zero if $p_{j,k}=0$. On the one hand, to ensure that $\rho_{AB}=\sum_{x,y=1}^m \ketbra{x}{y}\ox \rho_{xy}$ is positive semidefinite, one can verify that the $k$-th rows of $\rho_{j1},\rho_{j2},\cdots,\rho_{jm}$ are all filled with zeros, and the $k$-th columns of $\rho_{1j},\rho_{2j},\cdots,\rho_{mj}$ are all filled with zeros. On the other hand, to ensure that $\rho_{AB}^\Gamma=\sum_{x,y=1}^m \ketbra{y}{x}\ox \rho_{xy}$ is also positive semidefinite due to the assumption that $\rho_{AB}$ is a PPT state, one can verify that $k$-th columns of $\rho_{j1},\rho_{j2},\cdots,\rho_{jm}$ are all filled with zeros, and the $k$-th rows of $\rho_{1j},\rho_{2j},\cdots,\rho_{mj}$ are all filled with zeros. Then the assertion (ii.a) holds.

Second, we show the assertion (ii.b). If 
$$k\in\{k_1,\cdots,k_s\}\cup\{l_1,\cdots,l_t\},$$
then we have $p_{x,k}=0$ or $p_{y,k}=0$ by the assumption. Thus, according to the assertion (ii.a), we conclude that both the $k$-th rows and the $k$-columns of $\rho_{xy}$ and $\rho_{yx}$ are filled with zeros.

Third, we prove the assertion (ii.c). Since any permutation has a unitary representation, then we may assume that $\rho_{AB}^+=\sum_{j=1}^m p_{j,j}\proj{j,j}$ up to LU equivalence.
According to assertion (ii.a), one can verify that $\rho_{AB}$ has the same diagonal line as $\rho_{AB}^+$, and the non-zero elements of $\rho_{AB}$ which are not diagonal entries could only appear in the $(u,v)$ entries of $\rho_{xy}$ and the $(v,u)$ entries of $\rho_{yx}$, if $p_{x,u}\cdot p_{y,v}>0,~\forall 1\leq x,y\leq m$. Based on this observation, we conclude that the non-diagonal entries of $\rho_{AB}$ are all zeros, as $\rho_{AB}^\Gamma$ is also positive semidefinite. It follows that $\rho_{AB}=\rho_{AB}^+$. So $\rho_{AB}$ is a real and separable state.

This completes the proof.
\qed

\textbf{Proof of Lemma \ref{le:prs-prop}.}
(i) Suppose that $\alpha_{AB},\beta_{AB}$ both are PPT states in $\cP_{rs}(m,n)$. By definition for any $\tilde{\alpha}_{AB}\sim_{LU}\alpha_{AB}$ and any $\tilde{\beta}_{AB}\sim_{LU}\beta_{AB}$, the real parts of $\tilde{\alpha}_{AB}$ and $\tilde{\beta}_{AB}$ are always separable. 
First we show that $\cP_{rs}(m,n)$ is convex. One can verify that $t\alpha_{AB}+(1-t)\beta_{AB}$ is also a PPT state for any $t\in[0,1]$. Then, for any LU $U\ox V$ and any $t\in [0,1]$, we obtain that
\beq
\label{eq:prs-prop-1}
\begin{aligned}
&\left[(U\ox V)(t\alpha_{AB}+(1-t)\beta_{AB})(U\ox V)^\dg\right]^+ \\
=&t[(U\ox V)\alpha_{AB}(U\ox V)^\dg]^+ \\
&+(1-t)[(U\ox V)\beta_{AB}(U\ox V)^\dg]^+ \\
=&t\tilde{\alpha}_{AB}^+ + (1-t)\tilde{\beta}_{AB}^+ .
\end{aligned}
\eeq
Hence, we conclude that any state LU equivalent to $(t\alpha_{AB}+(1-t)\beta_{AB})$ has a separable real part, which implies that $(t\alpha_{AB}+(1-t)\beta_{AB})\in\cP_{rs}(m,n)$. It follows that $\cP_{rs}(m,n)$ is convex.
Second, we show that $\cP_{rs}(m,n)$ is a closed set. Denote by $\cD(m,n)$ the set of all bipartite states supported on $\bbC^m\ox\bbC^n$. It is known that $\cD(m,n)$ is closed and convex, and obviously $\cP_{r,s}(m,n)\subset\cD(m,n)$. Denote $\cD(m,n)\backslash\cP_{rs}(m,n)$ as $\cP^c_{rs}(m,n)$. Thus, the statement that $\cP_{r,s}(m,n)$ is closed is equivalent to that $\cP^c_{rs}(m,n)$ is an open set. Next, we shall equivalently show that $\cP^c_{rs}(m,n)$ is open. For any $\rho_{AB}\in\cP^c_{rs}(m,n)$, it follows from assertion (i) that $\rho_{AB}$ is entangled, and from the definition of $\cP_{rs}(m,n)$ that there exists an LU $U\ox V$ such that the real part of $(U\ox V)\rho_{AB}(U\ox V)^\dg$ is entangled. It follows from Theorem \ref{le:realEW} (ii) that $(U\ox V)\rho_{AB}(U\ox V)^\dg$ is detected by an REW, namely $W_r$. That is,
\begin{equation}
    \label{eq:qrs-closed-1}
\bal
&\tr\left((U\ox V)\rho_{AB}(U\ox V)^\dg W_r\right) \\
=&\tr\left(\rho_{AB}(U\ox V)^\dg W_r(U\ox V)\right)<0.
\eal
\end{equation}
First, since $\rho_{AB}$ is entangled, there exists a small enough neighborhood $U(\rho_{AB};\epsilon_1)$ such that any state in $U(\rho_{AB};\epsilon_1)$ is entangled. Note that $U(\rho_{AB};\epsilon_1)$ represents a set, any state in which is distant to $\rho_{AB}$ less than $\epsilon_1$ measured by the trace norm. Second, due to Eq. \eqref{eq:qrs-closed-1}, there exists another small enough neighborhood $U(\rho_{AB};\epsilon_2)$ such that for any $\alpha_{AB}\in U(\rho_{AB};\epsilon_2)$ the following inequality holds:
\begin{equation}
    \label{eq:qrs-closed-2}
\bal
&\tr\left((U\ox V)\alpha_{AB}(U\ox V)^\dg W_r\right) \\
=&\tr\left(\alpha_{AB}(U\ox V)^\dg W_r(U\ox V)\right)<0.
\eal
\end{equation}
It follows that the neighborhood $U(\rho_{AB};\epsilon_2)$ is included in $\cP^c_{rs}(m,n)$ by definition.
Let $\epsilon=\min\{\e_1,\e_2\}$. Therefore, we conclude that for any $\rho_{AB}\in\cP^c_{rs}(m,n)$, there always exists a small enough neighborhood $U(\rho_{AB};\epsilon)$ which is strictly included in $\cP^c_{rs}(m,n)$. It implies that the set $\cP^c_{rs}(m,n)$ is open by the knowledge of functional analysis, and equivalently the set $\cP_{r,s}(m,n)$ is closed.

(ii) Assume that $\rho_{AB}\in\cP_{rs}(m,n)$. The states of interest in the following assertions are naturally PPT states. It remains to consider whether the corresponding real parts are always separable.
First, one can verify that $\sigma_{AB}\sim\rho_{AB}$ is equivalent to $\sigma_{AB}^*\sim\rho_{AB}^*$. For any $\sigma^*_{AB}\sim\rho^*_{AB}$, we equivalently obtain that $\sigma_{AB}=(\sigma_{AB}^*)^*\sim\rho_{AB}$, and thus $\sigma_{AB}^+$ is separable by definition. Since $(\sigma^*)^+=\frac{1}{2}(\sigma^*+\sigma)=\sigma^+$, we obtain that $(\sigma_{AB}^*)^+$ is also separable as $\sigma_{AB}^+$ is separable. It implies that assertion (ii.a) holds by definition.
Second, assume that $\tilde{\rho}_{AB}\sim\rho_{AB}$. Then for any $\sigma_{AB}\sim\tilde{\rho}_{AB}$, it follows that $\sigma_{AB}\sim\rho_{AB}$. Since $\rho_{AB}\in\cP_{rs}(m,n)$, we obtain that $\sigma_{AB}^+$ is separable by definition. Due to $\sigma_{AB}\sim\tilde{\rho}_{AB}$, it implies that $\tilde{\rho}_{AB}\in\cP_{rs}(m,n)$, i.e. assertion (ii.b) holds. Third, for any $\sigma_{AB}\sim\rho_{AB}^\Gamma$, we have $\sigma_{AB}^\Gamma\sim\rho_{AB}$. Since $\rho_{AB}\in\cP_{rs}(m,n)$, by definition it follows that $(\sigma_{AB}^\Gamma)^+$ is separable. For any bipartite state $\alpha_{AB}$ we claim that $(\alpha_{AB}^\Gamma)^+=(\alpha_{AB}^+)^\Gamma$ for the following reason.
\beq
\label{eq:Prs-1}
\bal
(\alpha_{AB}^\Gamma)^+&=\frac{1}{2}(\alpha_{AB}^\Gamma+(\alpha_{AB}^\Gamma)^*)=\frac{1}{2}(\alpha_{AB}+\alpha_{AB}^*)^\Gamma \\
&=(\alpha_{AB}^+)^\Gamma.
\eal
\eeq
Due to Eq. \eqref{eq:Prs-1}, we conclude that $(\sigma_{AB}^+)^\Gamma$ is separable as $(\sigma_{AB}^\Gamma)^+$ is separable. Hence, we obtain that $\sigma_{AB}^+$ is separable. By definition it implies that assertion (ii.c) holds.
Fourth, $\rho_{BA}$ is obtained by applying a swap operation on $\rho_{AB}$, and thus $\rho_{BA}$ is supported on $\bbC^n\ox\bbC^m$. For any $\sigma_{BA}\sim\rho_{BA}$, we have $\sigma_{AB}\sim\rho_{AB}$. One can also verify that $\sigma_{BA}^+$ is from swapping two systems of $\sigma_{AB}^+$. By definition it follows from $\rho_{AB}\in\cP_{rs}(m,n)$ that $\sigma_{AB}^+$ is separable. Thus, we conclude that $\sigma_{BA}^+$ is separable for any $\sigma_{BA}\sim\rho_{BA}$. It means that $\rho_{BA}\in\cP_{rs}(n,m)$, i.e. assertion (ii.d) holds. 

This completes the proof.
\qed

\section{Proofs of some results in Sec. \ref{subsec:cj2}}
\label{sec:proof-3}

\textbf{Proof of Lemma \ref{le:cj2-c1}.}
(i) Suppose that $W$ is an $m\times n$ EW which can be projected to a $2\times n$ EW by a product matrix $X=U\ox V$ where $U$ is of Schmidt rank two. Then $XWX^\dg$ detects some entangled state $\alpha$ supported on $\bbC^2\ox\bbC^n$. According to the spectral decomposition of $\alpha$, we determine that $XWX^\dg$ must detect some pure entangled state $\ket{\psi}\in\bbC^2\ox\bbC^n$, where $\ket{\psi}\in\cR(\alpha)$. Let $\ket{\psi}=x_1\ket{a_1,b_1}+x_2\ket{a_2,b_2}$ be the Schmidt decomposition, where $x_1,x_2>0$. There exists a local unitary $U$ and the projector $P=(\proj{0}+\proj{1})\ox(\proj{0}+\proj{1})$ such that $PU\ket{\psi}=x_1\ket{0,0}+x_2\ket{1,1}$. Denote $\tilde{W}=PUXWX^\dg U^\dg P$. It follows that
\beq
\label{eq:cj2-1}
\bal
&(x_1\bra{0,0}+x_2\bra{1,1})\tilde{W}(x_1\ket{0,0}+x_2\ket{1,1})\\
=&\bra{\psi}XWX^\dg\ket{\psi}<0.
\eal
\eeq
Thus, the projected $\tilde{W}$ is a two-qubit EW which detects the two-qubit entangled pure state $x_1\ket{0,0}+x_2\ket{1,1}$.

(ii) Let $\rho_{AB}$ be an $m\times n$ and $1$-undistillable NPT state. By the definition of $1$-undistillable states, for any product matrix $X=U\ox V$ with Schmidt-rank-two matrix $U$, the projected state $X\rho_{AB}X^\dg$ can only be a $2\times n$ PPT state. Here we suppose that $W=\rho_{AB}^\G$ which is an EW according to Choi isomorphism. Then we obtain that
\beq
\label{eq:cj2-2}
(XWX^\dg)^\G\equiv (U^T\ox V)\rho_{AB} (U^T\ox V)^\dg
\eeq
must be a PPT state, as $\rho_{AB}$ is $1$-undistillable. Since $(XWX^\dg)^\G$ is a PPT state, it follows that $XWX^\dg$ is positive semidefinite and thus no longer an EW.

(iii) As we know, an EW detects PPT entangled states if and only if it is non-decomposable. We also claim that if an EW detects some PPT entangled state, then it must detect both PPT entangled and NPT states for the following reason. Suppose that a non-decomposable EW $W_P$ detects a PPT entangled state $\rho$. It follows that $\cR(\rho)$ is not spanned by product vectors, and $W_P$ detects at least one pure entangled state in $\cR(\rho)$ due to $\tr(W_P\rho)<0$. Any pure entangled state is NPT and thus $W$ detects both PPT entangled and NPT states.
Next, we prove the assertion by contradiction. Assume that $W$ can only detect NPT states and $XWX^\dg$ detects both PPT entangled and NPT states for some product matrix $X=U\ox V$. Let $W=P+Q^\G$ for some positive semidefinite $P,Q$. Then we obtain
\beq 
\label{eq:decewproj}
\bal
XWX^\dg&=XPX^\dg+XQ^\Gamma X^\dg\\
&=XPX^\dg+(U\ox V)Q^\Gamma(U^\dg\ox V^\dg)\\
&=XPX^\dg+[(U^*\ox V^\dg)Q(U^T\ox V)]^\Gamma.
\eal
\eeq
It follows that both $XPX^\dg$ and $(U^*\ox V^\dg)Q(U^T\ox V)$ are positive semidefinite, which implies that $XWX^\dg$ remains decomposable, and thus the projection $XWX^\dg$ can only detect NPT states. Then we obtain a contradiction, and conclude that $W$ should be non-decomposable and can detect both PPT entangled and NPT states.

This completes the proof. 
\qed

\textbf{Proof of Lemma \ref{le:projmindim}.}
Denote by $m(y^*),n(y^*)$ the two local dimensions of $\ket{y^*}$ and assume that $p\equiv \min\{m(y^*),n(y^*)\}=\min\limits_{\ket{y}}\{m(y),n(y)\}$. Then the Schmidt decomposition of $\ket{y^*}$ reads as $\ket{y^*}=\sum_{j=0}^{p-1}\la_j\ket{a_j,b_j}$ where $\forall \la_j>0$. There exists a local unitary operator $X=U\ox V$ and the projector $P=I_p\ox I_p$ such that $PX\ket{y^*}=\sum_{j=0}^{p-1}\la_j \ket{j,j}\in\mathbb C^p\ox\mathbb C^p$. 

(i) Since $\ket{y^*}$ is in the negative eigenspace of $\rho^\G$, we obtain that $\bra{y^*}\rho^\G\ket{y^*}<0$. It implies that $PX\rho^\G X^\dg P$ has a negative eigenvalue. It follows that 
\beq  
\label{eq:projmindim-1}
\left(P(X^\G)^\dg\rho X^\G P\right)^\G=PX\rho^\G X^\dg P.
\eeq
It implies that $\rho$ can be locally projected to an NPT state $P(X^\G)^\dg\rho X^\G P$ supported on $\mathbb C^p\ox\mathbb C^p$.

(ii) The proof is similar to that of assertion (i).  First, it obvious that $\bra{a,b}PXWX^\dg P\ket{a,b}\geq 0$ for any product state $\ket{a,b}$. Second, for the similar reason we conclude that $PXWX^\dg P$ has a negative eigenvalue as $W$ detects the pure entangled state $\ket{y^*}$. By definition $PXWX^\dg P$ is an EW supported on $\mathbb C^p\ox\mathbb C^p$.

This completes the proof.
\qed

\end{document}